\documentclass{article}      
\usepackage{graphicx, amsmath, bm}
\usepackage[left=3cm, right=3cm]{geometry}
\newtheorem{theorem}{Theorem}
\newtheorem{proof}{Proof}
\begin{document}

\title{Backward in time problem of a double porosity material with microtemperature}



\author{Olivia Ana Florea\\
Faculty of Mathematics andComputer Science, Transilvania University of Bra\c sov, Romania\\
{olivia.florea@unitbv.ro}}

\date{}

\maketitle

\begin{abstract}
In the present study we consider the theory of thermoelastodynamics in the case of materials with double porosity structure and microtemperature. This study is devoted to the investigation of a backward in time problem associated with double porous thermoelastic materials with microtemperature. In the first part of the paper, in case of the bounded domains the impossibility of time localization of solutions is obtained. This study is equivalent to the uniqueness of solutions for the backward in time problem. In the second part of the paper, a Phragmen-Lindelof alternative in the case of semi-infinite cylinders is obtained.\\
\textbf{keywords}{Thermoelasticity with double porosity \and Microtemperature \and Backward problem \and Impossibility of localization}
\end{abstract}

\section{Introduction}
In the last years many authors were interested in the linear theory of elastic materials with double porosity. The first studies regarding this theory are encountered in the papers of   Barenblatt, \cite{1}.  The concept of double porosity model allows for the body to have a double porous structure:
a  macro  porosity  connected  to  pores  in  the  body  and  a  micro  porosity  connected  to
fissures in the skeleton.  According to Barneblatt, \cite{2},  Berryman, \cite{3} and Khalili, \cite{6}, the particular applications of materials with double porosity are in  geophysics
and according to Cowin, \cite{4} in mechanics of bone.  The basic equations for elastic materials with double porosity involve the
displacement vector field, a pressure associated with the pores and a pressure associated
with the fissures [6-8].  We note that in the equilibrium theory the fluid
pressures become independent of the displacement vector field.

The theory for the behaviour of porous
solids in which the skeletal or matrix materials are elastic and the interstices are void
of material was studied by Nunziato and Cowin, \cite{7}.  The intended applications of this theory are to geological materials such
as rocks and soils and to manufactured porous materials such as ceramics and pressed
powders. Iesan and Quintanilla, \cite{5} used the Nunziato-
Cowin theory of materials with voids to derive a theory of thermoelastic solids which have a double porosity structure. In
contrast with the classical theory of elastic materials with double porosity, the porosity
structure in the case of equilibrium is influenced by the displacement field.
According to Quintanilla \cite{8} is proved the impossibility of the localization in time of the solutions of the linear thermoelasticity with voids. 

The study of backward in time problem is very important from the thermomechanical point of view because it offers information about the behavior of the system in the past using the information that we have at the present time. Usualy the Saint Venant' s principle is used for the spatial behavior of the solutions for partial differential equations. The studies regarding the spatial decay estimates were obtained for elliptic \cite{flavin_knops}, parabolic [13-14] and hyperbolic \cite {flavin-knops2} equations. The main aim of the spatial decay estimates is to model the perturbations on a side of the boundary that are damped for the points located at some distance from this side of the boundary. For this analysis it is necessary to use a semi-infinite cylinder whose finite end is perturbed and our goal is to identify the effects when the spatial variable increases. The harmonic vibrations in thermoelastic dynamics with double porosity structure for the backward in time problem was studied by Florea, \cite{mms}.

The phenomenon for which the mechanisms of dissipation are very strong, such that the solutions vanish after a finite time, is known as the localization in time of solutions. The impossibility of localization in time of solutions is an open problem because the proof of this concept exists only in some linear situations. In the particular case of the linear thermodynamics theory of visco-elastic solids with voids, the solutions decay can be controlled by some particular exponential or polynomial functions, [17-21]. The problem of the impossibility of localization of solutions was proved for the classical thermoelasticity with porous dissipation \cite{pamplona} and in the isotherm case with porous and elastic viscosity \cite{quinta}.

The aim of our paper is to show that in the case of thermoelasticity with double porosity structure and microtemperature the only solution that vanishes after a finite time is the null solution, when the mechanisms of dissipation are the double porous dissipation, the temperature and the microtemperature. Our obtaining results can be also compared with those obtained in [17-21]. In our paper we will give information regarding the upper bound for the solution decay. In the previous results, [17-20], the authors proved that after a small period of time the thermomechanical deformations are very small and they can be neglected. In our paper we will highlight that they are not null for any positive time. The present study represents a continuation of the research regarding the impossibility of localization in thermo-porous-elasticity with microtemperatures realized by Quintanilla, \cite{quint}, using the results of Florea, \cite{ijam}.

The present study is structured as follows: in the second section the basic equations for the backward in time problem in the case of materials with double porosity structure and microtemperature are described. Also, in this section the conditions imposed on the parameters that influence the behavior of the porous materials are presented. The impossibility of localization in time of solutions for the backward in time problem for a double porous material with microtemperature is expressed in the third section. We state here the conservation of the energy law and we highlight the main theorem of the present study. For the particular case of a semi-infinite cylinder a Phragmen-Lindelof alternative is obtained in the section 4. In the last section of the paper are drawn the conclusions of the present study.

\section{Basic equations for the double porous materials with microtemperature}
The equations of evolution that govern the problem of thermoelasticity with double porosity structure for the materials with microtemperature in the absence of the supply terms are, \cite{casas1}, \cite{casas2}:

\begin{flalign}
\label{m1}
t_{ji,j}&=\rho\ddot u_i\nonumber\\
\sigma_{j,j}+\xi&=k_1\ddot\phi\\
\tau_{j,j}+\zeta&=k_2\ddot\psi\nonumber
\end{flalign}
where: $\rho$ is the mass density, $k_1,k_2$ are the coefficients of equilibrated inertia, $\sigma_j,\tau_j$ are the equilibrated stress vectors, $\xi,\zeta$ are the intrinsic equilibrated body forces, $t_{ji}$ are the stress tensors, $u_i$ is the displacement, $\phi, \psi$ are the volume fraction fields in the reference configuration.

The equation of energy is given in \eqref{m2} and the equation of the first moment of energy is given in \eqref{m3}:
\begin{align}
\label{m2}
\rho T_0\dot\eta=Q_{j,j}
\end{align}
\begin{align}
\label{m3}
\rho\dot\varepsilon_i=Q_{ji,j}+Q_i-q_i
\end{align}
where where $T_0$ is the constant absolute temperature of the body in the reference configuration, $\eta$ is the entropy, $Q_j$ is the heat flux$, \varepsilon_i$ represent the first moment of energy vector and $q_i$ is the microheat flux average, $Q_{ji}$ is the first heat flux moment tensor.

We will consider in our study that we deal with a centrosymmetric material. In this case the constitutive equations for the linear theory are:
\begin{flalign}
\label{f4}
t_{ij} &= C_{ijkl}u_{k,l}+B_{ij}\phi+D_{ij}\psi-\beta_{ij}\theta \nonumber\\
\sigma_i &=\alpha_{ij}\phi_{,j}+b_{ij}\psi_{,j}- N_{ij}T_j\nonumber\\
\tau_i &= b_{ji}\phi_{,j}+\gamma_{ij}\psi_{,j}-M_{ij}T_j\nonumber\\
\xi &=-B_{ij}u_{i,j}-\alpha_1\phi-\alpha_3\psi+\gamma_1\theta\\
\zeta &=-D_{ij}u_{i,j}-\alpha_3\phi-\alpha_2\psi+\gamma_2\theta \nonumber\\ 
\rho\eta &=\beta_{ij}u_{i,j}+\gamma_1\phi+\gamma_2\psi+a\theta\nonumber\\
Q_i &= \kappa_{ij}\theta_{,j}+L_{ij}T_j\nonumber\\
\rho\varepsilon_i &=-N_{ij}\phi_{,j}-M_{ji}\psi_{,j}-P_{ij}T_j\nonumber\\
Q_{ij} &=-A_{ijrs}T_{s,r}\nonumber\\
q_i &=(L_{ij}-R_{ij})T_j+(\kappa_{ij}-\lambda_{ij})\theta_{,j}\nonumber
\end{flalign}
where $C_{ijkl}$ is the elasticity tensor, $\beta_{ij}$ is the thermal dilatation tensor,
$\kappa_{ij}$ is the heat conductivity tensor, $\beta_{ij}$ is the tensor of thermal dilatation, $B_{ij}, D_{ij}, \alpha_{ij}, b_{ij}, \gamma_{ij}, \alpha_1, \alpha_2, \alpha_3, \gamma_1, \gamma_2, a$ are typical functions in double porous theory and $N_{ij}, M_{ij}, R_{ij}, \lambda_{ij}, A_{ijrs}$ are tensors which are usual in the theories with microtemperatures. In the constitutive equations \eqref{f4} $\theta$ represents the temperature and $T_i$ are the microtemperatures.

Introducing the constitutive equations \eqref{f4} into the evolution equations \eqref{m1} the system of the field equations for the thermoelasticity with double porosity and microtemperatures is obtained: 
\begin{subequations}
\begin{equation}
\label{5a}
\rho\ddot u_i =\left(C_{jikl}u_{k,l}+B_{ji}\phi+D_{ij}\psi-\beta_{ij}\theta\right)_{,j}\tag{2.5.a}
\end{equation}
\begin{equation}
\label{5b}
k_1\ddot\phi=\left(\alpha_{ij}\phi_{,i}+b_{ij}\psi_{,i}-N_{ij}T_j\right)_{,j}-B_{ij}u_{i,j}-\alpha_1\phi-\alpha_3\psi+\gamma_1\theta\tag{2.5.b}
\end{equation}
\begin{equation}
\label{5c}
k_2\ddot\psi=\left(b_{ij}\phi_{,i}+\gamma_{ij}\psi_{,i}-M_{ij}T_j\right)_{,j}-D_{ij}u_{i,j}-\alpha_3\phi-\alpha_2\psi+\gamma_2\theta\tag{2.5.c}
\end{equation}
\begin{equation}
\label{5d*}
a\dot\theta =-\beta_{ij}\dot u_{i,j}-\gamma_1\dot\phi-\gamma_2\dot\psi+\frac{1}{T_0}\left(\kappa_{ij}\theta_{,j}+L_{ij}T_j\right)_{,j}\tag{2.5.d*}
\end{equation}
\begin{equation}
\label{5e*}
P_{ij}\dot T_j=\left(A_{ijrs}T_{s,r}\right)_{,j}-R_{ij}T_j-\lambda_{ij}\theta_{,j}-N_{ij}\dot\phi_{,j}-M_{ij}\dot\psi_{,j}\tag{2.5.e*}
\end{equation}
\end{subequations}
Proving the uniqueness of the solution of the backward in time problem, implies the impossibility of localization of the solutions of the above system. The system of equations which describes the backward in time problem is given by the same set of equations as \eqref{5a}-\eqref{5c} while \eqref{5d*} and \eqref{5e*} change into:
\begin{subequations}
\begin{equation}
\label{5d}
a\dot\theta =-\beta_{ij}\dot u_{i,j}-\gamma_1\dot\phi-\gamma_2\dot\psi-\frac{1}{T_0}\left(\kappa_{ij}\theta_{,j}+L_{ij}T_j\right)_{,j}\tag{2.5.d}
\end{equation}
\begin{equation}
\label{5e}
P_{ij}\dot T_j=-\left(A_{ijrs}T_{s,r}\right)_{,j}+R_{ij}T_j+\lambda_{ij}\theta_{,j}-N_{ij}\dot\phi_{,j}-M_{ij}\dot\psi_{,j}\tag{2.5.e}
\end{equation}
\end{subequations}

Because the constitutive coefficients are symmetric we have:
$$C_{ijkl}=C_{klij}; \alpha_{ij}=\alpha_{ji}; b_{ij}=b_{ji}; B_{ij}=B_{ji}, D_{ij}=D_{ji}.$$
For the case of anisotropic and homogeneous material we can draw the assumption that the tensors $A_{ijrs}, P_{ij}, N_{ij}, M_{ij}, L_{ij}, R_{ij}, \lambda_{ij}$ are also symmetric:
$$A_{ijkl}=A_{lkij}, P_{ij}=P_{ji}, M_{ij}=M_{ji}, L_{ij}=L_{ji}, N_{ij}=N_{ji}, R_{ij}=R_{ji}, \lambda_{ij}=\lambda_{ji}.$$
In the context of theories with microtemperature as a consequence of Clausius-Duhem inequality, we have the following assumption, \cite{casas1}:
\setcounter{equation}{5}
\begin{equation}
\label{f8}
\kappa_{ij}\theta_{,i}+\left(L_{ij}+T_0\lambda_{ij}\right)\theta_{,j}T_i+T_0R_{ij}T_iT_j+T_0A_{jirs}T_{i,j}T_{s,r}\geq 0
\end{equation}
In order to obtain the estimated results it is necessary to impose the positivity of several functions and tensors:
\begin{subequations}
\begin{flalign}
\label{a1}\tag{a.1}
\rho(X)\geq\rho_0>0; \hspace{1em} k_1(X)\geq k_0^1>0;\hspace{1em}
k_2(x)\geq k_0^2>0; \nonumber\\
\hspace{1em} a(x)\geq a_0>0; \hspace{1em} P_{ij}\xi_i\xi_j\geq p_0\xi_i\xi_i, p_0>0\nonumber
\end{flalign}
\begin{equation}
\label{a2}\tag{a.2}
\kappa_{ij}\xi_i\xi_j+(L_{ij}+T_0\lambda{ij})\xi_j\zeta_i+T_0R_{ij}\zeta_i\zeta_j\geq C_0(\xi_i\xi_i+\zeta_i\zeta_i), C_0>0, (\forall) \xi_i\zeta_i
\end{equation}
\begin{flalign}
C_{ijkl}u_{i,j}u_{k,l}+\alpha_{ij}\phi_{,i}\phi_{,j}+\gamma_{ij}\psi_{,i}\psi_{,j}+2b_{ij}\phi_{,i}\psi_{,j}+2B_{ij}u_{i,j}\phi+2D_{ij}u_{i,j}\psi
+\nonumber\\+\alpha_1\phi^2+\alpha_2\psi^2
+2\alpha_3\phi\psi
\label{a3}\tag{a.3}
\geq C^*\left(u_{i,j}u_{i,j}+\phi_{,i}\phi_{,i}+\psi_{,i}+\psi_{,i}+\phi^2+\psi^2\right), 
\end{flalign}
\begin{equation*}
\alpha_{ij}\xi_i\xi_j\geq0; \hspace{1em}b_{ij}\xi_i\xi_j\geq 0, (\forall) \xi_i
\end{equation*}
\begin{equation}
\label{a4}\tag{a.4}
A_{jirs}\xi_{ij}\xi_{sr}\geq C_1\xi_{ij}\xi_{ij}, C_1>0, (\forall) \xi_{ij}
\end{equation}
\end{subequations}

The assumption \eqref{a1} is related to the thermomechanical characteristics, \eqref{a2} and \eqref{a4} are consequences of the Clausius-Duhem inequality, \eqref{a3} gives the information that the internal energy is positive and may be expressed based on the theory of mechanical stability. 

\section{Main results regarding the impossibility of localization in time}

Let us consider a bounded domain $B$ with the boundary $\partial B$. The study of impossibility of localization in time for solutions of the backward in time problem is equivalent with the study of the uniqueness of solutions for the mentioned problem given by the system of equations \eqref{5a}-\eqref{5e}. To prove the uniqueness of solutions for the backward in time problem it is sufficient to show that  only the null solution satisfies our problem with null initial and boundary conditions. In the next computations we assume that the domain $B$ is smooth enough to apply the divergence theorem.

The initial conditions are:
\begin{flalign}
\label{e6}
u_i(\bm{X},0)=\dot u_i(\bm{X},0)=\phi(\bm{X},0)=\dot\phi(\bm{X},0)=0\\
\psi(\bm{X},0)=\dot\psi(\bm{X},0)=\theta(\bm{X},0)=0, \hspace{2em} T_i(\bm{X},0)=0\hspace{2em}\bm{X}\in B\nonumber
\end{flalign}
and the boundary conditions:

\begin{flalign}
\label{e7}
u_i(\bm{X},t)=\phi(\bm{X},t)=\psi(\bm{X},t)=\theta(\bm{X},t)=T_i(\bm{X},0)=0,\hspace{2em} \bm{X}\in\partial B, t\geq 0
\end{flalign}

The aim of this section is to obtain the nergy relation for the double porous material with microtemperature. We will multiply \eqref{5a} by $\dot u_i$, \eqref{5b} by $\dot\phi$, \eqref{5c} by $\dot\psi$, \eqref{5d} by $\theta$ and \eqref{5e} by $T_j$, the obtained relations will be integrated on $[0,t]$ and they will be summed. Using the divergence theorem and taking into account the boundary conditions based on the principle of conservation of energy, we have the following relation:
\begin{flalign}
\label{f11}
E_1(t)=&\frac{1}{2}\int\limits_B\left(\rho\dot u_i\dot u_i+k_1\dot\phi^2+ k_2\dot\psi^2+ a\theta^2+ P_{ij}T_jT_j+C_{ijkl}u_{i,j}u_{k,l}+2B_{ij}\phi u_{i,j} +\right.\nonumber\\
&+\left. 2D_{ij}\psi u_{i,j}+\alpha_{ij}\phi_{,j}\phi_{,j}+\gamma_{ij}\psi_{,j}\psi_{,j}+2b_{ij}\psi_{,j}\phi_{,j}+\alpha_1\phi^2+2\alpha_3\phi\psi+\alpha_2\psi^2\right)dV=\\
&=\int\limits_0^t\int\limits_B\left[\frac{1}{T_0}\left(\kappa_{ij}\theta_{,i}\theta_{,j}+L_{ij}T_i\theta_{,j}\right)+A_{ijrs}T_{s,r}T_{i,j}+R_{ij}T_iT_j+\lambda_{ij}\theta_{,j}T_i\right]dVds\nonumber
\end{flalign}
Using the same procedure of multiplying the equations \eqref{5a} by $\dot u_i$, \eqref{5b} by $\dot\phi$, \eqref{5c} by $\dot\psi$, \eqref{5d} by $-\theta$ and \eqref{5e} by $-T_j$, integrating on $[0,t]$ and using the divergence theorem we have the following expression:
\begin{flalign}
\label{f12}
E_2(t)=\frac{1}{2}\int\limits_B\left(\rho\dot u_i\dot u_i+k_1\phi^2+k_2\psi^2-a\theta^2-P_{ij}T_jT_j+C_{ijkl}u_{k,l}u_{i,j}+2B_{ij}\phi u_{i,j}+2D_{ij}\psi u_{i,j}\right.+\nonumber\\
\left.+\alpha_{ij}\phi_{,j}\phi_{,j}+2b_{ij}\psi_{,j}\phi_{,j}+\alpha_1\phi^22\alpha_3\phi\psi+\gamma_{ij}\psi_{,j}\psi_{,j}+\alpha_2\psi^2\right)dV=\\
=-\int\limits_{0}^t\left[\frac{1}{T_0}\int\limits_B\left(\kappa_{ij}\theta_{,j}\theta_{,j}+L_{ij}T_j\theta_{,j}+A_{ijrs}T_{s,r}T_{j,i}+R_{ij}T_jT_j+\lambda_{ij}\theta_{,j}T_j\right)dV\right]ds+\nonumber\\
+\int\limits_0^t\int\limits_B\left[\left(\beta_{ij}\theta\right)_{,j}\dot u_i-\left(M_{ij}T_j\right)\dot\psi-\left(N_{ij}T_j\right)\dot\phi+\gamma_1\theta\dot\phi+\gamma_2\theta\dot\psi\right]dVds\nonumber
\end{flalign}
Taking into consideration the equations \eqref{5a}-\eqref{5e}, the initial and boundary conditions, \eqref{e6}, \eqref{e7}, the following identity is obtained:
\begin{flalign}
\label{f13}
\int\limits_B\left(\rho\dot u_i\dot u_i+k_1\dot\phi^2+k_2\dot\psi^2-a\theta^2-P_{ij}T_jT_j\right)dV&=\int\limits_B\left(C_{ijkl}u_{i,j}u_{k,l}+\alpha_{ij}\phi_{,i}\phi_{,j}+\gamma_{ij}\psi_{,i}\psi_{,j}+2b_{ij}\psi_{,j}\phi_{,i}+\right.\nonumber\\
&\left.+2B_{ij}u_{i,j}\phi+2D_{ij}u_{i,j}\psi+\alpha_1\phi^2+\alpha_2\psi^2+2\alpha_3\phi\psi\right)dV
\end{flalign}

The impossibility of localization of the solutions in the theory with double porosity and microtemperature is proved in the following theorem.
\begin{theorem}
Let $(u_i,\phi,\psi, \theta, T_i)$ be a solution of the backward in time problem \eqref{5a}-\eqref{5e} with the initial conditions \eqref{e6} and the boundary conditions \eqref{e7}. The only solution of the mentioned problem is the null solution $u_i=0,$ $\phi=0,$ $\psi=0,$ $\theta=0,$ $T_i=0$.
\end{theorem}
\begin{proof}
Replacing  \eqref{f13} into \eqref{f12} we obtain a new expression for $E_2(t)$:
\begin{flalign*}
E_2(t)=\int\limits_B\left(C_{ijkl}u_{i,j}u_{k,l}+\alpha_{ij}\phi_{,i}\phi_{,j}+\gamma_{ij}\psi_{,i}\psi_{,j}+2b_{ij}\phi_{,i}\psi_{,j}+2B_{ij}u_{i,j}\phi+\right.\\
+\left. 2D_{ij}u_{i,j}\psi+\alpha_1\phi^2+\alpha_2\psi^2+2\alpha_3\phi\psi\right)dV=\\
=-\int\limits_0^t\int\limits_B\frac{1}{T_0}\left(\kappa_{ij}\theta_{,i}\theta_{,j}+L_{ij}\theta_{,j}T_i+T_0A_{ijrs}T_{s,r}T_{j,i}+T_0R_{ij}T_iT_j+T_0\lambda_{ij}\theta_{,j}T_i\right)dVds+\\
+\int\limits_0^t\int\limits_B\left[(\beta_{ij}\theta)_{,j}-(M_{ij}T_i)_{,j}\dot\psi-(N_{ij}T_i)_{,j}\dot\phi+\gamma_1\theta\dot\phi+\gamma_2\theta\dot\psi\right]dVds
\end{flalign*}
The energy can be wxpressed under the bellow form, if we consider a positive constant $\varepsilon$, small enough:
$$E(t)=E_2(t)+\varepsilon E_1(t), \hspace{2em} \varepsilon\in (0,1)$$
Taking into account that $E(t)$ is a positive function we have the following form for the energy:
\begin{flalign*}
E(t)&=\frac{\varepsilon}{2}\int\limits_B\left(\rho\dot u_i\dot u_i+k_1\dot\phi^2+k_2\dot\psi^2+a\theta^2+P_{ij}T_iT_j\right)dV+\\
+&\frac{2+\varepsilon}{2}\int\limits_B\left(C_{ijkl}u_{i,j}u_{k,l}+\alpha_{i,j}\phi_{,i}\phi_{,j}+\gamma_{ij}\psi_{,i}\psi_{,j}+2b_{ij}\phi_{,i}\psi_{,j}+2B_{ij}u_{i,j}\phi+2D_{ij}u_{i,j}\psi+\right.\\
+&\left.\alpha_1\phi^2+\alpha^2\psi^2+2\alpha^3\phi\psi \right)dV
\end{flalign*}
On the other hand,
\begin{flalign*}
E(t)=&-\int\limits_0^t\int\limits_B\frac{1}{T_0}\left(\kappa_{ij}\theta_{,i}\theta_{,j}+L_{ij}\theta_{,j}T_i+T_0A_{ijrs}T_{s,r}T_{j,i}+T_0R_{ij}T_iT_j+T_0\lambda_{ij}\theta_{,j}T_i\right)dVds+\\
+&\int\limits_0^t\int\limits_B\left[(\beta_{ij}\theta)_{,j}\dot u_i-(M_{ij}T_i)_{,j}\dot\psi-(N_{ij}T_i)_{,j}\dot\phi+\gamma_1\theta\dot\phi+\gamma_2\theta\dot\psi\right]dVds+\\
+&\varepsilon\int\limits_0^t\int\limits_B\frac{1}{T_0}\left(\kappa_{ij}\theta_{,i}\theta_{,j}+L_{ij}T_i\theta_{,j}+T_0A_{ijrs}T_{s,r}T_{j,i}+T_0R_{ij}T_iT_j+T_0\lambda_{ij}T_i\theta_{,j}\right)dVds
\end{flalign*}
The above relation yields, for $\varepsilon\in(0,1)$:
\begin{flalign*}
E(t)=&-(1-\varepsilon)\int\limits_0^t\int\limits_B\frac{1}{T_0}\left(\kappa_{ij}\theta_{,i}\theta_{,j}+L_{ij}\theta_{,j}T_i+T_0A_{ijrs}T_{s,r}T_{j,i}+T_0R_{ij}T_iT_j+T_0\lambda_{ij}T_i\theta_{,j}\right)dVds+\\
+&\int\limits_0^t\int\limits_B\left[ (\beta_{ij}\theta)_{,j}\dot u_i-(M_{ij}T_i)_{,j}\dot\psi-(N_{ij}T_i)_{,j}\dot\phi+\gamma_1\theta\dot\phi+\gamma_2\theta\dot\psi\right]dVds
\end{flalign*}
from where:
\begin{flalign*}
\frac{dE(t)}{dt}=&-(1-\varepsilon)\int\limits_B\frac{1}{T_0}\left(\kappa_{ij}\theta_{,i}\theta_{,j}++L_{ij}\theta_{,j}T_i+T_0A_{ijrs}T_{s,r}T_{j,i}+T_0R_{ij}T_iT_j+T_0\lambda_{ij}T_i\theta_{,j}\right)dVds+\\
+&\int\limits_B\left[ (\beta_{ij}\theta)_{,j}\dot u_i-(M_{ij}T_i)_{,j}\dot\psi-(N_{ij}T_i)_{,j}\dot\phi+\gamma_1\theta\dot\phi+\gamma_2\theta\dot\psi\right]dVds
\end{flalign*}
but,
$$\int\limits_B \left(\beta_{ij}\theta\right)_{,j}\dot u_idV=\int\limits_B\beta_{ij,j}\theta\dot u_i dV+\int\limits_B\beta_{ij}\theta_{,j}\dot u_idV$$
The inequality of arithmetic and geometric means implies that:
$$\int\limits_B\left(\beta_{ij}\theta\right)_{,j}\dot u_idV\leq C_1\int\limits_B\left(\rho\dot u_i\dot u_i+a\theta^2\right)dV+\varepsilon_1\int\limits_B\kappa_{ij}\theta_{,i}\theta_{,j}dV$$
where $\varepsilon_1$ is small enough, $C_1$ is a positive constant that can be determined based on the constitutive coefficients and $\varepsilon_1$;
$$\int\limits_B(M_{ij}T_i){,j}\dot\psi dV\leq C_2\int\limits_B\left(k_2\dot\phi^2+P_{ij}T_iT_j\right)dV$$
where $C_2$ can be determined. Therefore, there is a positive constant $C$ such that:
$$\frac{dE}{dt}\leq C\int\limits_B\left(\rho\dot u_i\dot u_i+k_1\dot\phi^2+k_2\dot\psi^2+a\theta^2+P_{ij}T_iT_j\right)dV$$
which is equivalent with the estimate:
$$\frac{dE}{dt}\leq C^* E(t)\Leftrightarrow \frac{dE}{E}\leq C^* dt\Leftrightarrow \ln E\leq C^*t+\mathcal{C}\Leftrightarrow E(t)\leq \mathcal{C}e^{C^*t}.$$
For $t=0$ we will have the estimate: $$E(t)\leq E(0)e^{C^*t}$$
But, the initial condition leads us to $E(t)=0$ for every $t\geq 0$ that is equivalent with:
$$\dot u_i=0; \dot\phi=0; \dot\psi=0;\theta=0; T_i(t)=0 \Leftrightarrow u_i=C_1; \phi=C_2;\psi=C_3;\theta=T_i=0$$
taking into account the initial conditions \eqref{e6} we obtain that the solution for our problem is the null solution:
$$u_i=0;\phi=0;\psi=0; \theta=0; T_i=0$$
\end{proof}
\section{Phragmen-Lindelof alternative for the solution of backward in time problem with double porosity and microtemperature}

We consider a semi-infinite prismatic cylinder $B=D\times[0,\infty)$ that is occupied by a body with a double porosity structure with micro-temperature. By $D$ we note the cross section in the cylinder. The boundary of the section is a piece-wise continuously differentiable curve denoted by $\partial D$ sufficiently smooth to admit application of divergence theorem. The lateral surface of the cylinder is $\Pi=\partial D\times(0,\infty)$. The cylinder is assumed to be free of load on the lateral boundary surface.\\
The lateral boundary conditions are:
\begin{flalign}
\label{e15}
u_i(\textbf{X},t)=0; \phi(\textbf{X},t)=0; \psi(\textbf{X},t)=0; \theta(\textbf{X},t)=0; T_i(\textbf{X}, T)=0 (\textbf{X},t)\in\Pi\times(0,\infty)
\end{flalign}
On the base of the cylinder the following boundary conditions are assumed:
\begin{flalign}
\label{e16}
u_i(x_1,x_2,0,t)=\tilde u_i;\phi(x_1,x_2,0,t)=\tilde\phi;\nonumber\\
\psi(x_1,x_2,0,t)=\tilde\psi;\theta(x_1,x_2,0,t)=\tilde\theta; T_i(x_1,x_2,0,t)=\tilde T_i
\end{flalign}
For the solution of the problem determined by the system \eqref{5a}-\eqref{5e} with initial conditions \eqref{e15} and boundary conditions \eqref{e16} we want to obtain a Phragmen-Lindelof alternative necessary for the interpretation of the behavior of the solution of our boundary value problem. Our aim in this section is to estimate the absolute value of the defined function $H_\omega$ from \eqref{e17} by means of its spatial derivative.

We define the function:
\begin{flalign}
\label{e17}
H_\omega(z,t)&=\int\limits_0^t\int\limits_{D(z)}e^{-2\omega s}\left[C_{i3kl}u_{k,l}+B_{i3}\phi+D_{i3}\psi-\beta_{3i}\theta\right]\dot u_i dads+\nonumber\\
&+\int\limits_0^t\int\limits_{D(z)}e^{-2\omega s}\left[\alpha_{i3}\phi_{,i}+b_{i3}\psi_{,i}-N_{i3}T_i\right]\dot\phi dads+\\
&+\int\limits_0^t\int\limits_{D(z)}e^{-2\omega s}\left[b_{3i}\phi_{,i}+\gamma_{i3}\psi_{,i}-M_{i3}T_i\right]\dot\psi dads+\nonumber\\
&+\int\limits_0^t\int\limits_{D(z)}e^{-2\omega s}\frac{1}{T_0}\left[\kappa_{i3}\theta_{,i}+L_{i3}T_i\right]\theta dads\nonumber\\
&+\int\limits_0^t\int\limits_{D(z)}e^{-2\omega s}\left(A_{3irs}T_{s,r}+R_{i3}T_i+\lambda_{ij}\theta_{,i}\right)T_i dads\nonumber
\end{flalign}
Here we have $D(z)=\{\textbf{X}\in B|x_3=z\}$ that denotes the cross section of the cylinder at a distance $z$ from the base.
Through means of the divergence theorem and employing the field equations, boundary and initial conditions we obtain:
\begin{flalign}
\label{e18}
H_\omega(z+h,t)-H_\omega(z,t)=\frac{1}{2}\int\limits_{R(z+h,z)}\chi_\omega(t)dV, (\forall)h>0
\end{flalign}
where $R(z+h,z)=\{\textbf{X}\in B|z<x_3<z+h\}$.

The internal energy is:
\begin{flalign}
\label{20}
\Phi &=\rho\dot u_i\dot u_i+k_1\dot\phi^2+k_2\dot\psi^2+a\theta^2+P_{ij}T_jT_j+C_{ijkl}u_{i,j}u_{k,l}++2B_{ij}u_{i,j}\phi+2D_{ij}u_{i,j}\psi+\\
&+\alpha_{ij}\phi_{,i}\phi_{,j}+\gamma_{ij}\psi_{,i}\psi_{,j}+2b_{ij}\phi_{,i}\phi_{,j}+\alpha_1\phi^2+\alpha_2\psi^2+2\alpha_3\phi\psi\nonumber
\end{flalign}
such that:
\begin{flalign}
\label{e19}
\chi_\omega(t)=e^{-2\omega t}\Phi(t)+\int\limits_0^t e^{-2\omega s}\left[2\omega\Phi(s)+2\frac{\kappa_{ij}}{T_0}\theta_{,i}(s)\theta_{,j}(s)+2\frac{L_{ij}}{T_0}\theta_{,i}T_j(s)\right.\\
+\left. 2A_{ijrs}T_{s,r}(s)T_{i,j}(s)+2R_{ij}T_i(s)T_j(s)+2\lambda_{ij}\theta_{,i}T_j(s)\right]ds\nonumber
\end{flalign}
From \eqref{e18} we have:
$$\frac{\partial H_\omega}{\partial z}=\frac{1}{2}\int\limits_{D(z)}\chi_\omega(t)dz$$
that leads to the following relation:
\begin{flalign}
\label{e21}
\frac{\partial H_\omega}{\partial z}=\frac{e^{-2\omega t}}{2}\int\limits_{D(z)}\Phi(t)da+\int\limits_0^t\int\limits_{D(z)}e^{-2\omega s}\left[\omega\Phi(s)+W\right]dads
\end{flalign}
where,
\begin{flalign*}
W= \frac{\kappa_{ij}}{T_0}\theta_{,i}\theta_{,j}+\frac{L_{ij}}{T_0}\theta_{,i}T_j+A_{ijrs}T_{s,r}T_{i,j}+R_{ij}T_iT_j+2\lambda_{ij}\theta_{,i}T_j
\end{flalign*}
Further, we want to estimate the absolute value of $H_\omega$ in terms of spatial derivatives, in order to get a differential inequality, such that:
\begin{flalign}
\label{e22}
|H_\omega|\leq C_\omega\frac{\partial H_\omega}{\partial z}, \hspace{1em}(\forall)z\geq 0
\end{flalign}
The above inequality is known in the literature of specialty reagarding the spatial estimate as Phragm\'en-Lindel\"of alternative.

Under the assumption\eqref{a3} the internal energy from \eqref{20} leads us to the following inequality:
\begin{flalign*}
\Phi&\geq \rho\dot u_i\dot u_i+k_1\dot\phi^2+k_2\dot \psi^2+a\theta^2+P_{ij}T_iT_i+C^*\left(u_{i,j}u_{i,i}+\phi_{,i}\phi_{,i}+\psi_{,i}\psi_{,i}+\phi^2+\psi^2\right)
\end{flalign*}
Therefore the relation \eqref{e18} yields:
\begin{flalign*}
&H_\omega(z+h,t)-H_\omega(z,t)\geq\\
&\geq\frac{1}{2}\int\limits_R e^{-2\omega s}\left[\rho\dot u_i\dot u_i+k_1\dot\phi^2+k_2\dot \psi^2+a\theta^2+P_{ij}T_iT_i+C^*\left(u_{i,j}u_{i,i}+\phi_{,i}\phi_{,i}+\psi_{,i}\psi_{,i}+\phi^2+\psi^2\right)\right]da+\\
&+\int\limits_0^t\int\limits_R e^{-2\omega s}\left\{\omega\left[k_1\dot\phi^2+k_2\dot \psi^2+a\theta^2+P_{ij}T_iT_i+C^*\left(u_{i,j}u_{i,i}+\phi_{,i}\phi_{,i}+\psi_{,i}\psi_{,i}+\phi^2+\psi^2\right)\right]+\kappa_{ij}\theta_{,i}\theta_{,i}\right.\\
&+\left.\frac{\kappa_{ij}}{T_0}\theta_{,i}\theta_{,i}+\frac{L_{ij}}{T_0}\theta_{,i}T_i+A_{ijrs}T_{i,j}T_{i,j}+R_{ij}T_iT_i+\lambda_{ij}\theta_{,i}T_i\right\}dads
\end{flalign*}
Based on the inequality of arithmetic and geometric means and also the Cauchy-Schwarz inequality we obtain:
\begin{flalign*}
|H_\omega(z,t)|&\leq C_\omega\left[\frac{e^{-2\omega t}}{2}\int\limits_{D(z)}\Phi(t)dz+\int\limits_0^t\int\limits_{D(z)}\omega e^{-2\omega s}\Phi(s) da ds+\right.\\
&+\int\limits_0^t\int\limits_{D(z)}e^{-2\omega s}\left[\frac{1}{T_0}\left(\kappa_{ij}\theta_{,i}\theta_{,j}+L_{ij}T_i\theta_{,j}\right)\right.+\left. A_{ijrs}T_{i,j}T_{r,s}+R_{ij}T_iT_j+\lambda_{ij}T_i\theta_{,j}\right]dads
\end{flalign*} 
Thus the alternative \eqref{e22} was proved.

From the inequality \eqref{e22} we can extract the following two inequalities:
\begin{flalign}
\label{e23}
-\frac{\partial H_\omega}{\partial z}\leq \frac{1}{C_\omega}H_\omega
\text{ and } 
\frac{\partial H_\omega}{\partial z}\geq \frac{1}{C_\omega}H_\omega
\end{flalign}
Taking into consideration the computations from Flavin, \cite{flavin-knops2} we obtain two estimates:
\begin{flalign}
\label{e25}
H_\omega(z,t)\geq H_\omega(z_0,t)e^{\frac{z-z_0}{C_\omega}}
\end{flalign}
$(\forall) z\geq z_0, z_0>0$ and $H_\omega(z_0,t)>0$ that lead to: $\lim\limits_{z\rightarrow \infty} e^{-\frac{z}{C_\infty}}\int\limits_{R(z)}\chi_\omega(t)dv>0$ and
\begin{flalign}
\label{f29}
-H_\omega(z,t)\leq H_\omega(z_0,t)e^{\frac{-z}{C_\omega}}
\end{flalign}
$(\forall) z\geq 0$ and $H_\omega(z,t)\leq 0$. From \eqref{f29} it is obvious that $H_\omega(z,t)\rightarrow0$ for $z\rightarrow\infty$.

Let us introduce the following estimate:
\begin{flalign}
\label{f30}
E_\omega(z,t)= \frac{e^{-2\omega t}}{2}\int\limits_{R(z))}\Phi(t)dz+\int\limits_0^t\int\limits_{R(z)} e^{-2\omega s}
\left[\omega\Phi(s)+\frac{1}{T_0}\left(\kappa_{ij}\theta_{,i}\theta_{,j}+L_{ij}T_i\theta_{,j}\right)\right.\\
+\left.A_{ijrs}T_{i,j}T_{r,s}+R_{ij}T_iT_j+\lambda_{ij}T_i\theta_{,j}\nonumber \right]da ds 
\end{flalign}
where $R(z)=\{\textbf{X}\in B|z<x_3\}$. Based on \eqref{f29} we observe that:
\begin{flalign}
\label{f31}
E_\omega(z,t)\leq E_\omega(0,t)e^{-\frac{z}{C_\omega}}, z\geq 0
\end{flalign}
Now, we can draw the following conclusions: if $(u_i, \phi,\psi,\theta,T_i)$ is a solution of the backward in time problem defined by the system \eqref{5a}-\eqref{5e} with the null initial conditions \eqref{e6} and boundary conditions \eqref{e7} there are two situations: the solution satisfies the asymptotic condition: $\lim\limits_{z\rightarrow\infty}e^{-\frac{z}{C_\omega}}\int\limits_{R(z)}\chi_\omega(t)dv>0$ pr it satisfies the decay estimate \eqref{f29}. 
This study can continue with obtaining of the upper bound for the amplitude $E_\omega(0,t)$ in terms of the boundary conditions, but this analysis will be the subject of another paper.

\section{Conclusions}
In the present paper it was studied the impossibility of localization in time for the solutions of the boundary value problem associated with the linear thermoelastic materials with double porosity structure and microtemperature. The uniqueness of the solutions for the backward in time problem in case of the materials with double porosity structure with microtemperature was proved. We can draw the conclusion that for the backward in time problem the only solution that vanishes is the null solution for every $t>0$. In the case of linear thermoelastic theories this results can not certify that the thermomechanical deformations from double porous bodies with microtemperature vanish after a finite time. In this situation it is necessary that the time should be unbounded to guarantee that the fraction of volumes becomes the same as the reference configuration. We obtained a function that defines a measure on the solutions and we deduced the usual exponential type alternative for the solutions of the problem defined in a semi-infinite cylinder.

\end{document}